\documentclass[
]{ceurart1}

\usepackage{listings}
\usepackage{amsthm}
\newtheorem{theorem}{Theorem}
\newtheorem{definition}{Definition}

\newcommand{\avec}[1]{\boldsymbol{#1}}

\newcommand{\bool}{\textit{bool}}

\newcommand{\var}{\textit{var}}

\newcommand{\q}{{\avec{q}}}

\newcommand{\sig}{\mathit{sig}}

\newcommand{\no}{$\mathsf{no}$}

\newtheorem{lemma}[theorem]{Lemma}

\usepackage{misc}

\newcommand{\FO}{\ensuremath{\mathsf{FO}}}

\newcommand{\FOT}{\ensuremath{\mathsf{FO^2}}}
\newcommand{\CT}{\ensuremath{\mathsf{C^2}}}
\newcommand{\FOTOE}{\ensuremath{\mathsf{FO^2 1E}}}
\newcommand{\FOTTE}{\ensuremath{\mathsf{FO^2 2E}}}
\newcommand{\GFTTEc}{\ensuremath{\mathsf{GF^2 2E_c}}}

\renewcommand{\ALC}{\ensuremath{\mathcal{ALC}}}
\newcommand{\ALCTE}{\ensuremath{\mathcal{ALC}^{\cap,\neg,\textit{id}}}\mathsf{2E}}

\renewcommand{\P}{\mathcal{P}}
\newcommand{\chv}{\overline{\boldsymbol{v}}}
\newcommand{\chw}{\overline{\boldsymbol{w}}}

\renewcommand{\abf}{\ensuremath{\boldsymbol{a}}\xspace}
\renewcommand{\bbf}{\ensuremath{\boldsymbol{b}}\xspace}
\renewcommand{\xbf}{\ensuremath{\boldsymbol{x}}\xspace}
\renewcommand{\ybf}{\ensuremath{\boldsymbol{y}}\xspace}

\newcommand{\bis}{\boldsymbol{\beta}}

\begin{document}



\title{Interpolant Existence is Undecidable for Two-Variable First-Order Logic with Two Equivalence Relations}



\author[1]{Frank Wolter}[%
email=wolter@liverpool.ac.uk,
]
\address[1]{University of Liverpool, Ashton Street, Liverpool L69 3BX, UK}

\author[2]{Michael Zakharyaschev}[%
email=m.zakharyaschev@bbk.ac.uk,
]
\address[3]{Birkbeck, University of London, Malet Street, London WC1E 7HX, UK}

\begin{abstract}
The interpolant existence problem (IEP) for a logic $L$ is to decide, given formulas $\varphi$ and $\psi$, whether there exists a formula $\iota$, built from the shared symbols of $\varphi$ and $\psi$, such that $\varphi$ entails $\iota$ and $\iota$ entails $\psi$ in $L$. If $L$ enjoys the Craig interpolation property (CIP), then the IEP reduces to validity in $L$. Recently, the IEP has been studied for logics without the CIP. The results obtained so far indicate that even though the IEP can be computationally harder than validity, it is decidable when $L$ is decidable. Here, we give the first examples of decidable fragments of first-order logic for which the IEP is undecidable. Namely, we show that the IEP is undecidable for the two-variable fragment with two equivalence relations and for the two-variable guarded fragment with individual constants and two equivalence relations. We also determine the corresponding decidable Boolean description logics for which the IEP is undecidable.  
\end{abstract}



\maketitle

\section{Introduction}\label{intro}

A first-order (\FO) formula $\iota$ is called a Craig interpolant for \FO-formulas $\varphi$ and $\psi$ if $\varphi\models \iota \models \psi$ and $\iota$ is built from the shared non-logical symbols of $\varphi$ and $\psi$. Interpolants have been applied in many areas ranging from formal  verification~\cite{DBLP:conf/cav/McMillan03} and software specification~\cite{diaconescu1993logical} to theory combinations~\cite{DBLP:conf/cade/CimattiGS09,CalEtAl20,DBLP:journals/jar/CalvaneseGGMR22} and query reformulation and rewriting in databases~\cite{DBLP:series/synthesis/2011Toman,DBLP:series/synthesis/2016Benedikt}. 

For many fragments $L$ of \FO, including many description logics (DLs), the existence of a Craig interpolant for $\varphi$ and $\psi$ is guaranteed if $\varphi$ entails $\psi$.  This phenomenon is known as  the Craig interpolation property (CIP) of $L$. The CIP holds, for instance, for the DLs $\mathcal{ALC}$, $\mathcal{ALCI}$, and $\mathcal{ALCQI}$~\cite{TenEtAl13}, and also for the guarded negation fragment of \FO~\cite{DBLP:journals/tocl/BenediktCB16}. In this case, the problem to decide whether an interpolant
for $\varphi$ and $\psi$ exists reduces to showing that $\varphi$ entails $\psi$ in $L$, and so is not harder than entailment. Moreover, interpolants can often be extracted from a proof of the entailment. 

An important consequence of the CIP of a logic $L$ is the projective Beth definability property (BDP) of $L$: if a relation is implicitly definable over a signature $\varrho$ in $L$, then it is explicitly definable over $\varrho$ in $L$. The BDP is used in ontology engineering for extracting explicit definitions of concepts or nominals from ontologies~\cite{TenEtAl06,TenEtAl13,ArtEtAl21}, developing and maintaining ontology alignments~\cite{DBLP:conf/ekaw/GeletaPT16}, and  for robust  modularisations and decompositions of ontologies~\cite{DBLP:series/lncs/KonevLWW09,DBLP:conf/rweb/BotoevaKLRWZ16}. Explicit definitions are used in ontology-based data management to equivalently rewrite ontology-mediated
queries~\cite{DBLP:conf/ijcai/SeylanFB09,TomanWed20,FraEtAl13,FraKer19,TomWed21}.

Unfortunately, there are many prominent fragments of \FO{} that do not enjoy the CIP and BDP: for example, DLs with nominals such as $\mathcal{ALCO}$ and/or role inclusions such as $\mathcal{ALCH}$~\cite{TenEtAl13}, the guarded and the two-variable fragment of \FO~\cite{comer1969,Pigozzi71,DBLP:journals/ndjfl/MarxA98}.  
To extend interpolation-based techniques to such formalisms, it has recently been suggested to investigate the interpolant existence problem (IEP) for logics $L$ without the CIP: given formulas $\varphi$ and $\psi$ as an input, decide whether they have a Craig interpolant in $L$. It has been shown that the IEP is decidable for many DLs with nominals and role inclusions~\cite{DBLP:journals/tocl/ArtaleJMOW23}, Horn-DLs~\cite{DBLP:journals/corr/abs-2202-07186}, the two-variable and guarded fragments of \FO~\cite{DBLP:conf/lics/JungW21}, and some decidable fragments of first-order modal logics~\cite{DBLP:conf/kr/KuruczWZ23}. These positive results naturally lead to a conjecture that the IEP for $L$ is decidable whenever the entailment in $L$ is decidable.  

In this paper, we disprove this conjecture by showing that the IEP is undecidable for the two-variable \FO{} with two equivalence relations and also for the two-variable guarded fragment of \FO{} with constants and two equivalence relations. In the realm of Description Logic, the former result means that interpolant existence is undecidable for concept inclusions in $\ALC$ enriched with Boolean operators on roles and the identity role. As a consequence, we also obtain the undecidability of the explicit definition existence problem (EDEP) for each of these logics $L$, which is the problem of deciding, given formulas $\varphi$, $\psi$ and a signature $\varrho$, whether there exists an explicit $\varrho$-definition of $\psi$ modulo $\varphi$ in $L$. (It is known~\cite{DBLP:conf/kr/KuruczWZ23} that the IEP and EDEP are polynomially reducible to each other.) 

The paper is structured as follows. Section~\ref{prelims} introduces the fragments of \FO{} we are interested in and summarises the necessary technical tools and results. Section~\ref{undec} proves the undecidability of the IEP for the two-variable \FO{} with two equivalence relations by reduction of the infinite Post Correspondence Problem. Section~\ref{other} shows how this result can be adapted to the two-variable guarded fragment with constants and two equivalence relations and to suitable description logics. Finally, Section~\ref{open} discusses a few challenging open problems.



\section{Preliminaries}\label{prelims}


Let $\sigma$ be a signature of individual constants and \emph{unary and binary} predicate symbols, including two distinguished binary \emph{equivalence predicates} $E_1$, $E_2$. Let $\var$ be a set comprising two individual variables. 
We consider the following fragments of first-order logic $\FO(\sigma)$ over $\sigma$:
\begin{description}
\item[$\FOTTE(\sigma)$] is the set of \emph{constant-free} $\FO$-formulas constructed in the usual way from atoms $x=y$ with $x,y \in \var$ and $R(\xbf)$, for any $n$-ary predicate symbol $R \in \sigma$ and  $n$-tuple $\xbf$ of variables from $\var$\,;


\item[$\GFTTEc(\sigma)$] admits individual constants in atoms but restricts quantification to the patterns 
$$ 
\forall
\ybf\, \big( \alpha(\xbf,\ybf)\rightarrow \varphi(\xbf,\ybf) \big), \qquad 
\exists \ybf \, \big( \alpha(\xbf,\ybf)\wedge \varphi(\xbf,\ybf) \big), 
$$
where $\varphi(\xbf,\ybf)$ is a $\GFTTEc$-formula and $\alpha(\xbf,\ybf)$ is an atom with occurrences of $\xbf,\ybf$.

\end{description}
The \emph{signature} $\sig(\psi)$ of a formula $\psi$ in any of these logics is the set of predicate and constant symbols occurring in $\psi$. 
Formulas are interpreted in $\sigma$-\emph{structures} $\Amf= \big( \text{dom}(\Amf),(R^{\Amf})_{R\in \sigma},(c^{\Amf})_{c\in \sigma}) \big)$ with a domain $\text{dom}(\Amf) \ne \emptyset$, relations $R^{\Amf}$ on $\text{dom}(\Amf)$ of the same arity as predicates $R \in \sigma$, and $c^{\Amf} \in \text{dom}(\Amf)$. It is always assumed that $E^{\Amf}_1$ and $E^{\Amf}_2$ are \emph{equivalence relations} on $\text{dom}(\Amf)$. 
A \emph{pointed structure} is a pair $\Amf,\abf$ with a tuple $\abf=(a_1,\dots,a_n)$, $n \le 2$, of elements from $\text{dom}(\Amf)$.  

The decision problems for $\FOTTE$ and $\GFTTEc$ are known to be \coTwoNExpTime- and \TwoExpTime-complete, respectively; see~\cite{DBLP:journals/siamcomp/KieronskiMPT14,Pratt23book} and references therein.

\begin{definition}\em
Let $L \in \{\FOTTE, \GFTTEc\}$, let $\varphi(\xbf)$ and $\psi(\xbf)$ be $L(\sigma)$-formulas with the same free variables $\xbf$, and let $\varrho = \sig(\varphi) \cap \sig(\psi)$. An $L(\varrho)$-formula $\iota(\xbf)$ is called an $L$-\emph{interpolant} for $\varphi$ and $\psi$ if $\varphi(\xbf) \models \iota(\xbf)$ and $\iota(\xbf) \models \psi(\xbf)$. 
\end{definition}

Our concern here is the following \emph{$L$-interpolant existence problem} ($L$-\emph{IEP}, for short):
\begin{description}
\item[($\boldsymbol L$-\emph{IEP})] given $L$-formulas $\varphi(\xbf)$ and $\psi(\xbf)$, decide whether they have an $L$-interpolant $\iota(\xbf)$.
\end{description}

We remind the reader of a well-known criterion of interpolant existence in terms of appropriate bisimulations. Let $\varrho \subseteq \sigma$. Given pointed $\sigma$-structures $\Amf,\abf$ and $\Bmf,\bbf$, we call $\Amf,\abf$ and $\Bmf,\bbf$ \emph{$L(\varrho)$-equivalent} and write $\Amf,\abf \equiv_{L, \varrho} \Bmf,\bbf$ if $\Amf\models \varphi(\abf)$ iff $\Bmf\models \varphi(\bbf)$, for all $L(\varrho)$-formulas $\varphi$.
%
%
\begin{definition}\label{deftwo}\em 
A relation $\bis \subseteq \text{dom}(\Amf)\times\text{dom}(\Bmf)$ is an \emph{$\FOTTE(\varrho)$-bisimulation between \Amf and \Bmf} 
if $\bis$ is \emph{global} in the sense that $\text{dom}(\Amf)\subseteq \{a\mid (a,b)\in \bis\}$ and $\text{dom}(\Bmf)\subseteq \{b\mid (a,b)\in \bis\}$ and the following conditions are satisfied for all $(a,b)\in \bis$:
\begin{enumerate}
\item for every $a'\in\text{dom}(\Amf)$, there is a $b'\in \text{dom}(\Bmf)$ such that $(a',b')\in \bis$ and $(a,a')\mapsto (b,b')$ is a partial  $\varrho$-isomorphism between \Amf and \Bmf;
	
\item for every $b'\in\text{dom}(\Bmf)$, there is a $a'\in \text{dom}(\Amf)$ such that $(a',b')\in \bis$ and $(a,a')\mapsto (b,b')$ is a partial  $\varrho$-isomorphism between \Amf and \Bmf.
\end{enumerate}
\end{definition}

\begin{definition}\em 
A global relation $\bis\subseteq \text{dom}(\Amf)\times\text{dom}(\Bmf)$ is a \emph{$\GFTTEc(\varrho)$-bisimulation between \Amf and \Bmf}
if $(c^{\Amf},c^{\Bmf})\in \bis$ for all $c\in \varrho$ and the following conditions are satisfied for all $(a,b)\in \bis$:
\begin{enumerate}
\item for every $a'\in\text{dom}(\Amf)$ such that either $R^{\Amf}(a,a')$, for some binary $R\in\varrho$, or $a'=a$, condition 1 of Definition~\ref{deftwo} holds;
	\item for every $b'\in\text{dom}(\Bmf)$ such that either $R^{\Bmf}(b,b')$, for some binary $R\in\varrho$, or $b'=b$, condition 2 of Definition~\ref{deftwo} holds.
\end{enumerate}
\end{definition} 

For tuples $\abf=(a_1,\dots,a_n)$ and $\bbf=(b_1,\dots,b_n)$ with $n \le 2$, we write $\Amf,\abf \sim_{L,\varrho} \Bmf,\bbf$ if 
$\abf\mapsto\bbf$ is a partial $\varrho$-isomorphism between \Amf and
\Bmf and there is an $L(\varrho)$-bisimulation $\bis$ between $\Amf$ and
\Bmf such that $(a_i,b_i)\in \bis$, for all $i\leq n$.
The following characterisation is well known; see, e.g.,~\cite{goranko20075,DBLP:books/daglib/p/Gradel014,Pratt23book}: 

\begin{lemma}\label{lem:guardedbisim} 
Let $L \in \{\FOTTE, \GFTTEc\}$ and $\varrho \subseteq \sigma$. For any pointed $\sigma$-structures $\Amf,\abf$ and $\Bmf,\bbf$ with $|\abf| = |\bbf|$,
$$
\Amf,\abf \sim_{L,\varrho} \Bmf,\bbf \quad \text{ implies } \quad
\Amf,\abf \equiv_{L,\varrho} \Bmf,\bbf
$$
and, conversely, if structures $\Amf$ and $\Bmf$ are $\omega$-saturated, then
$$
\Amf,\abf \equiv_{L,\varrho} \Bmf,\bbf \quad \text{ implies } \quad 
\Amf,\abf \sim_{L,\varrho} \Bmf,\bbf.
$$
\end{lemma}
 
Using this lemma, one can obtain the following variant of Robinson's joint consistency criterion~\cite{Robinson56,DBLP:conf/lics/JungW21}:

\begin{lemma}\label{lemcharinterpolation}
Let $L \in \{\FOTTE, \GFTTEc\}$. Then $L(\sigma)$-formulas $\varphi(\xbf)$ and $\psi(\xbf)$ do not have an $L$-interpolant iff there exist pointed $L(\sigma)$-structures $\Amf,\abf$ and
$\Bmf,\bbf$ such that 

\begin{itemize}
\item $\Amf\models \varphi(\abf)$\textup{;}

\item $\Bmf \models \neg \psi(\bbf)$\textup{;}

\item $\Amf,\abf \sim_{L,\varrho} \Bmf,\bbf$, where $\varrho = \sig(\varphi) \cap \sig(\psi)$.
\end{itemize}
\end{lemma}

\begin{definition}\em
Given formulas $\varphi(\xbf)$, $\psi(\xbf)$ and a signature $\varrho$, 
an \emph{explicit $\varrho$-definition of $\psi(\xbf)$ modulo $\varphi(\xbf)$ in a logic $L$} is an $L(\varrho)$-formula $\chi(\xbf)$ such that $\models \varphi(\xbf) \to (\psi(\xbf) \leftrightarrow \chi(\xbf))$. 
\end{definition}

The \emph{explicit $\varrho$-definition existence problem} (\emph{EDEP}) for $L$ is formulated as follows:
\begin{description}
\item[($\boldsymbol L$-\emph{EDEP})] given $L$-formulas $\varphi(\xbf)$,  $\psi(\xbf)$ and a signature $\varrho$, decide whether there exists an explicit $\varrho$-definition of $\psi(\xbf)$ modulo $\varphi(\xbf)$ in $L$.
\end{description}

The IEP and EDEP turn out to be closely related~\cite{MGabbay2005-MGAIAD}. Here, we only need the following lemma whose proof can be found in~\cite{DBLP:conf/kr/KuruczWZ23}:

\begin{lemma}\label{p:cipvspbdp}
For any $L \in \{\FOTTE, \GFTTEc\}$, the $L$-IEP is polynomially reducible to the $L$-EDEP.
\end{lemma}


\section{Undecidability of the $\FOTTE$-IEP}\label{undec}

In this section, we prove the undecidability of the $\FOTTE$-IEP by a reduction of the undecidable \emph{infinite Post Correspondence Problem} ($\omega$PCP)~\cite{DBLP:journals/eik/Ruohonen85,DBLP:journals/ipl/Gire86}, which is formulated as follows: given an alphabet $\Gamma=\{A_{1},\dots,A_{n}\}$, $n \ge 2$, and a finite set $\P$ of pairs $(v_{1},w_{1}),\dots,(v_{k},w_{k})$ of non-empty words over $\Gamma$, decide whether there exists an infinite sequence of indices $i_{1},i_{2},\dots$, for $1 \le i_j \le k$ and $j < \omega$, such that the $\omega$-words $v_{i_{1}}v_{i_{2}}\ldots$ and $w_{i_{1}}w_{i_{2}}\ldots$ coincide. If this is the case, the $\omega$PCP instance $\P$ is said to \emph{have a solution}.

Suppose an $\omega$PCP instance $\P$ is given. Our aim is to construct $\FOTTE$-formulas $\varphi(x)$ and $\psi(x)$ such that $\P$ has a solution iff $\varphi(x)$ and $\neg\psi(x)$ are satisfied in $\FOTTE(\rho)$-bisimilar pointed structures, where $\varrho$ is the shared signature of $\varphi$ and $\psi$, and then apply Lemma~\ref{lemcharinterpolation}. In our construction 
$\varrho = \{R,S\} \cup \Gamma$ 
with two binary predicates $R$, $S$ and the $A_i \in \Gamma$ treated as unary predicates. 
The formula $\varphi(x)$ is defined by taking:  
\begin{align*}
\varphi(x) = &\ X_{0}(x) \wedge \exists y\, \big(R(x,y) \wedge X_{1}(y) \big)   \wedge \forall x\, \big[ X_{1}(x) \rightarrow \exists y\, \big( R(x,y) \wedge X_{2}(y)\big) \big] \wedge {}\\
& \ \forall x\, \big[ X_{2}(x) \rightarrow \exists y\, \big( R(x,y) \wedge X_{0}(y)\big) \big] \land \forall x,y\, \big((X_{0}(x) \wedge X_{0}(y)) \rightarrow x=y \big) \land{} \\
& \mbox{}\hspace*{6.5cm} \exists y \, \big[ S(x,y) \wedge \bigwedge_{A_i\in \Gamma} \big( A_i(y) \to A_i(y) \big)\big].
\end{align*}
The only purpose of $\varphi(x)$ is to generate an $R$-cycle with three points: for any structure $\Amf,a_0$, if $\Amf\models \varphi(a_0)$, then $\Amf$ contains a cycle $R(a_{0},a_{1}),R(a_1,a_{2}),R(a_2,a_0)$. The last conjunct of $\varphi$ is only needed to ensure that $\sig(\varphi) \cap \sig(\psi) = \varrho$.

Before defining $\psi(x)$ formally, we explain the intuition behind $\neg \psi(x)$. Suppose $\Amf,a_0$ and $\Bmf,b_0$ are such that $\Amf\models \varphi(a_0)$, $\Bmf \models \neg \psi(b_0)$ and $\Amf,a_0 \sim_{L,\varrho} \Bmf,b_0$, for $L = \FOTTE$. Then $\Bmf$ contains an $R$-chain $b_{0}, b_{1}, b_2, b_3,\dots$ such that $\Amf,a_{0} \sim_{L,\varrho} \Bmf,b_{3m}$, for $m < \omega$. The purpose of $\neg \psi(x)$ is to generate an infinite $S$-chain $\chv$ starting from $b_0$, along which an infinite sequence of words $v_j$ is written, and also an infinite $S$-chain $\chw$ starting from $b_3$, along which a sequence of $w_j$ is written. Using the equivalence relations $E_1$ and $E_2$, the formula $\neg\psi(x)$ ensures that the pairs $(v_j,w_j)$ are all in $\P$. That the resulting $\omega$-words coincide (and give a solution to $\P$) is ensured by $\Amf,a_0 \sim_{L,\varrho} \Bmf,b_0$. The intended models $\Amf,a_0$ and $\Bmf,b_0$ are illustrated in Fig.~\ref{intended}.

\begin{figure}[htb]
\begin{center}
\includegraphics[scale=0.53]{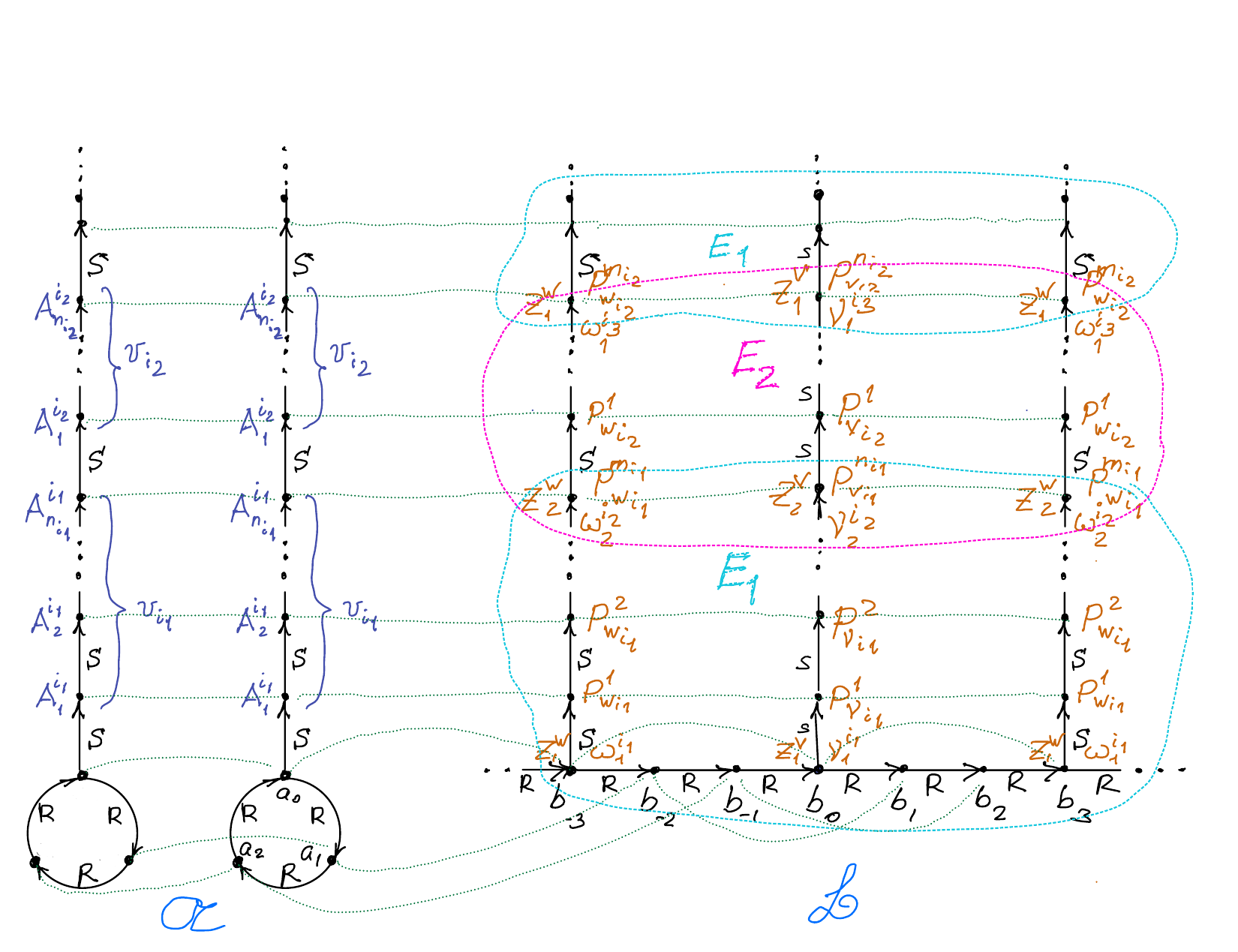}
\caption{Intended models $\Amf,a_0$ and $\Bmf,b_0$ with $\Amf\models \varphi(a_0)$, $\Bmf \models \neg \psi(b_0)$, and $\Amf,a_0 \sim_{L,\varrho} \Bmf,b_0$, where $\varrho$-bisimilar points in the disjoint union of $\Amf$ and $\Bmf$ are connected by dotted lines and $E_i$-equivalence classes in $\Bmf$, $i=1,2$, are encircled by dashed lines. Model $\Amf$ contains two disjoint copies of a 3-point $R$-cycle followed by an infinite $S$-chain. Model $\Bmf$ has an infinite $R$-chain of points $b_i$, for $i \in \mathbb Z$; each $b_{6i}$ starts an infinite $S$-chain for the word $v_{i_1}v_{i_2}\dots$, and  each $b_{6i+3}$ starts an infinite $S$-chain for the word $w_{i_1}w_{i_2}\dots$. These $\omega$-words coincide with the $\omega$-word written on the $S$-chains in $\Amf$.}
\label{intended}
\end{center}
\end{figure}

To define $\neg \psi$ formally, suppose that $v_{j}=A_{1}^{j} \dots A_{n_{j}}^{j}$ and $w_{j} = B_{1}^{j} \dots B_{m_{j}}^{j}$, for $1 \le j \le k$. 
Along with the $A_i \in \Gamma$, to generate unique $S$-chains $\chv$ and $\chw$, we require auxiliary unary predicates $P^l_{v_j}$, $l = 1, \dots, n_j$, and $P^l_{w_j}$, $l = 1, \dots, m_j$, as well as $Z^v_i$ and $Z^w_i$, for $i=1,2$, and $Z^v$, $Z^w$. Let $\bar{1}=2$, $\bar{2}=1$, $\bar{x}=y$, and $\bar{y}=x$. 

Now, we define $\neg \psi(x)$ to be a conjunction of the following four groups of formulas:

\begin{description}
\item[(generation)] $Z^v_1(x) \land Z^v(x)$, \ $\forall x,y\, \big(Z^u(x) \land S(x,y) \to Z^u(y) \big)$, for $u \in \{v,w\}$,\\
\mbox{}\hspace*{1.34cm} $\forall x\, \big( Z^v_i(x) \rightarrow \bigvee_{j\leq k} \nu_i^{j}(x) \big)$, \ $\forall x\, \big( Z^w_i(x) \rightarrow \bigvee_{j\leq k} \omega_i^{j}(x) \big)$, for $i=1,2$,\\
\mbox{}\hspace*{1.34cm} $\forall y\, \big( R(x,y) \rightarrow Y_1(y) \big)$,  $\forall x,y\, \big( Y_1(x) \land R(x,y) \rightarrow Y_2(y) \big)$,\\ 
\mbox{}\hfill$\forall x,y\, \big(Y_2(x) \land R(x,y) \rightarrow Z^w_1(y) \land Z^w(y) \big)$,
\end{description}
where, the $\nu_i^j$ generating $v_j$, for $i=1,2$, is defined recursively as 
\begin{align*}\label{four}
& \nu_i^j(x) = \alpha_{i,0}^j(x) \land \beta_{i,0}^j(x),\\
& \alpha_{i,l}^j(x) = \exists \bar x\, ( S(x,\bar x) \land E_i(x,\bar x) \land A_l^{j}(\bar x) \land P^l_{v_j}(\bar x) \land \alpha_{i,l+1}^j(\bar x) \big), \ \ \text{ for } l = 0,\dots,n_j -1, \\
& \alpha_{i,n_j}^j(x) = \exists \bar x\, \big( S(x,\bar x) \land E_i(x,\bar x) \land A_{n_j}^{j}(\bar x) \land P^{n_j}_{v_j}(\bar x) \land Z^v_{\bar{i}}(\bar x) \big), \\
& \beta_{i,l}^j(x) = \forall \bar x\, \big( S(x,\bar x) \to E_i(x,\bar x) \land A_l^{j}(\bar x) \land P^l_{v_j}(\bar x) \land \beta_{i,l+1}^j(\bar x) \big), \ \ \text{ for } l = 0,\dots,n_j -1, \\
& \beta_{i,n_j}^j(x) = \forall \bar x\, \big( S(x,\bar x) \to E_i(x,\bar x) \land A_{n_j}^{j}(\bar x) \land P^{n_j}_{v_j}(\bar x) \land Z^v_{\bar{i}}(\bar x) \big) ,
\end{align*}
and the $\omega_i^j$, generating $w_j$, are defined analogously but with $B^j_l$, $P^l_{w_j}$ and $m_j$ in place of $A^j_l$, $P^l_{v_j}$ and $n_j$, respectively. 
\begin{description}
\item[(disjointness)] $\forall x\, \big( P(x) \rightarrow \neg P'(x)\big)$,\ for any distinct $P,P' \in \Gamma$, distinct $P,P'$ of the form $P^l_{v_j}$, $P^l_{w_j}$, as well as for $\{P,P'\} = \{Z^u_1,Z^u_2\}$ with $u \in \{v,w\}$, and $\{P,P'\} = \{Z^v,Z^w\}$. 

\item[(coordination)] $\forall x,y\, \big( R(x,y) \rightarrow E_{1}(x,y) \big)$,\\
\mbox{}\hspace*{1.7cm} $\forall x \,\big[\big( Z^v_{i}(x) \wedge \nu_i^{j}(x)\big)  \rightarrow{}  \forall y\, \big(E_{i}(x,y)\wedge Z^w_{i}(y)\rightarrow \omega_i^{j}(y)\big) \big]$, for $i=1,2$, $j \le k$,\\
\mbox{}\hspace*{1.7cm} $\forall x,y\, \big(	E_{i}(x,y) \wedge Z^v_{\bar{i}}(x) \wedge Z^w_{\bar{i}}(y) \rightarrow E_{\bar{i}}(x,y) \big)$, for $i=1,2$.

\item[(uniqueness)] $\forall x,y \, \big( P(x) \land P(y) \land E_i(x,y) \to x=y \big)$, for   $P$ of the form $P^l_{v_j}$, $P^l_{w_j}$, $i=1,2$,\\
\mbox{}\hspace*{1.4cm} $\forall x,y \, \big( P^l_{u_j}(x) \land P^{l'}_{u'_{j'}}(y) \land E_i(x,y) \to \bot \big)$, for $i=1,2$ and all pairs $P^l_{u_j}$, $P^{l'}_{u'_{j'}}$\\
\mbox{}\hspace*{1.4cm} of the form $P^l_{v_j}$, $P^l_{w_j}$with $l \ne l'$.

%
\end{description}
%

We now show that the constructed $\FOTTE$-formulas $\varphi$ and $\psi$ are as required:

\begin{theorem}\label{oPCP}
Let $\varphi$ and $\psi$ be the $\FOTTE$-formulas defined above for a given $\omega{}$PCP instance $\P$. Then $\P$ has a solution iff $\varphi$ and $\neg\psi$ are satisfied in $\FOTTE(\varrho)$-bisimilar pointed structures. 
\end{theorem}
\begin{proof}
$(\Rightarrow)$ Suppose $v_{i_{1}}v_{i_{2}}\dots  = w_{i_{1}}w_{i_{2}}\dots$ is a solution to the given $\P$. It is not hard to check that, for $\Amf$ and $\Bmf$ in Fig.~\ref{intended}, we have $\Amf\models \varphi(a_0)$, $\Bmf \models \neg \psi(b_0)$, and $\Amf,a_0 \sim_{\FOTTE(\varrho)} \Bmf,b_0$. We only note that two disjoint copies of a 3-point $R$-cycle followed by an $S$-chain in $\Amf$ are needed to ensure $\FOTTE(\varrho)$-bisimilarity between $\Amf,a_0$ and $\Bmf,b_0$. Indeed, suppose, for example, that $a$ and $b$ are the $i$th points in the $S$-chains starting from $a_0$ and $b_0$, respectively. Now, if we consider the $(i+1)$th point $b'$ in the $S$-chains starting from $b_3$, then, to obtain a partial $\varrho$-isomorphism between $(b,b'$) and  some $(a,a')$, we may need to take the $(i+1)$th point $a'$ from the second $S$-chain in $\Amf$.

$(\Leftarrow)$ Suppose $\Amf\models \varphi(a_0)$, $\Bmf \models \neg \psi(b_0)$, and $\Amf,a_0 \sim_{\FOTTE,\varrho} \Bmf,b_0$. 
Consider first $\Bmf,b_0$. As $\Bmf \models \neg \psi(b_0)$, the axioms in the first two lines of \textbf{(generation)} give an infinite sequence $\chv = \chv_{i_1} \chv_{i_2} \dots$ depicted below. We show that the $\omega$-word $v_{i_1}v_{i_2} \dots$ over $\Gamma$ written on $\chv$ is unique. Indeed, by \textbf{(disjointness)} and \textbf{(generation)}, all $S$-successors of\\
\centerline{\includegraphics[scale=0.8]{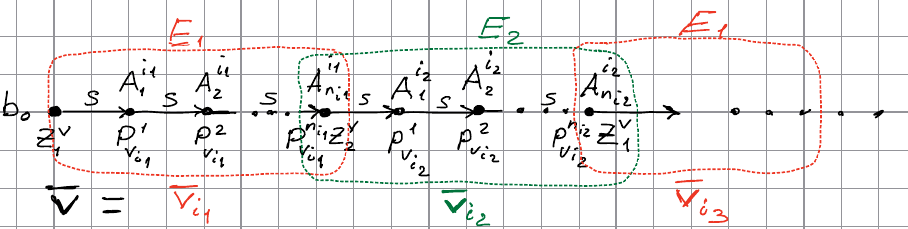}}\\
$b_0$ carry the same $P^1_{v_{i_1}}$, all of them are in the same $E_1$-equivalence class, and so must coincide by \textbf{(uniqueness)}. Next we see that all of the $S$-successors of the $P^1_{v_{i_1}}$-point  carry $P^2_{v_{i_1}}$, and so must coincide, etc. This gives us the word $v_{i_1}$ carried by the segment $\chv_{i_1}$ sitting in the same $E_1$-equivalence class. The last point in this segment carries $Z^v_2$, which triggers the next segment $\chv_{i_2}$ sitting in the same $E_2$-equivalence class, and so on. All of the points in $\chv$ carry $Z^v$ by the second formula in \textbf{(generation)}. Points on the boundaries between $\chv_{i_j}$ and $\chv_{i_{j+1}}$ belong to both $E_1$- and $E_2$-equivalence classes.

As $\Amf,a_0 \sim_{L,\varrho} \Bmf,b_0$, there are points $b_i$, $i >0$,  in $\Bmf$ with $R(b_0,b_1),R(b_1,b_2),R(b_2,b_3)$ and $b_3 \sim_{L,\varrho} b_0$. By \textbf{(generation)}, we have $Z^w_1(b_3)$, $Z^w(b_3)$, and so, by \textbf{(disjointness)}, $b_0 \ne b_3$. By the first formula in \textbf{(coordination)}, we must have $E_1(b_0,b_3)$.

The same argument as for $b_0$ above gives an $S$-chain $\chw$ starting from $b_3$ and disjoint from $\chv$ that defines uniquely an $\omega$-word $w_{i'_{1}}\dots w_{i'_{N}}\dots$ over $\Gamma$. The chain $\chw$ consists of consecutive segments $\chw_{i'_1},\chw_{i'_2},\dots$ such that $\chw_{i'_1}$ sits in the same $E_1$-equivalence class as $\chw_{i_1}$, $\chw_{i'_2}$ sits in an $E_2$-equivalence class, $\chw_{i'_3}$ in an $E_1$-equivalence class, etc.

By the second formula in \textbf{(coordination)}, we have $i_1 = i'_1$, and so $(v_{i_1},w_{i'_1}) \in \P$; by the third one, the ends of $\chv_{i_1}$ and $\chw_{i_1}$ are $E_2$-related, from which $i_2 = i'_2$ and $(v_{i_2},w_{i'_2}) \in \P$. The ends of $\chv_{i_2}$ and $\chw_{i_2}$ are $E_1$-related, so $i_3 = i'_3$ and $(v_{i_3},w_{i'_3}) \in \P$, etc. 

Finally, let $\chv$ be $S(b_0,c_0),S(c_0,c_1),\dots$ and let $\chw$ be $S(b_3,d_0),S(d_0,d_1),\dots$. Since these $S$-chains are unique and $b_3 \sim_{L,\varrho} b_0$, we must have $c_i \sim_{L,\varrho} d_i$, for all $i < \omega$. It follows that the $\Gamma$-symbols carried by each pair $c_i$ and $d_i$ must coincide, which gives a solution to $\P$.
\end{proof}

Since the $\omega$PCP is undecidable, as an immediate consequence of Theorem~\ref{oPCP} and Lemma~\ref{p:cipvspbdp} we obtain our main result:
\begin{theorem}
The $\FOTTE$-IEP and $\FOTTE$-EDEP are both undecidable.
\end{theorem}


\section{Undecidability of the IEP for Guarded and Description Logics}\label{other}

It is not hard to tweak the construction in the previous section to show  that the $\GFTTEc$-IEP is undecidable. Observe that the formula $\psi$ defined above can be equivalently represented as a formula in $\GFTTEc$ and that the $\FOTTE(\varrho)$-bisimulation between $\Amf$ and $\Bmf$ is actually a $\GFTTEc(\varrho)$-bisimulation (in this case, it is actually enough to have one copy of the $R$-cycle with an $S$-chain in $\Amf$). The only unguarded conjunct of $\varphi$ is $\forall x,y\, \big((X_{0}(x) \wedge X_{0}(y)) \rightarrow x=y \big)$, which is needed to generate an $R$-cycle with three points. The same result can be achieved using a guarded $\varphi'$ with two individual constants $c_1$ and $c_2$ as follows:
$$
\varphi'(x) = R(x,c_{1}) \wedge R(c_{1},c_{2}) \wedge R(c_{2},x) \wedge \exists y \, \big[ S(x,y) \wedge \bigwedge_{A_i\in \Gamma} \big( A_i(y) \to A_i(y) \big)\big].
$$
We then obtain Theorem~\ref{oPCP} with $\varphi'$ and $\GFTTEc$ in place of $\varphi$ and $\FOTTE$, respectively, and, as an immediate consequence, the following:

\begin{theorem}
The $\GFTTEc$-IEP and $\GFTTEc$-EDEP are both undecidable. 
\end{theorem}

These undecidability results can be interpreted as results about description logics with Boolean operators on roles that correspond to the two-variable fragments~\cite{DBLP:conf/dlog/LutzSW01,DBLP:conf/csl/LutzSW01}. The crucial operators required are role intersection, negation, and the identity role. Denote by $\mathcal{ALC}^{\cap,\neg,\text{id}}$2E the extension of the basic description logic $\mathcal{ALC}$ with the operators $\cap$ and $\neg$ on roles, the identity role id interpreted as $\{(a,a) \mid a\in \text{dom}(\Amf)\}$ in any structure $\Amf$, and two equivalence relations. The \emph{interpolant existence problem} for $\ALCTE$ (or $\ALCTE$-IEP, for short) is the problem to decide, given  $\ALCTE$-concepts $C_{1}$ and $C_{2}$, whether there exists an $\ALCTE$-concept $C$ built from the shared concept and role names from $C_{1}$ and $C_{2}$ such that $C_{1}\sqsubseteq C$ and $C \sqsubseteq C_{2}$ are valid concept inclusions. We then obtain:

\begin{theorem}
The $\ALCTE$-IEP and $\ALCTE$-EDEP are both undecidable. 
\end{theorem}

To obtain a DL version of the undecidability for the $\GFTTEc$-IEP, we can add nominals and the universal role to $\ALCTE$ and, to reflect guarded quantification, admit only Boolean combinations of roles that contain at least one positive occurrence of a role name (different from the universal role). 


\section{Open Problems}\label{open}

This paper presents first examples of decidable fragments of first-order logic, for which the interpolant existence problem and the explicit definition existence problem are undecidable. Numerous questions remain open, including the following:
\begin{itemize}
\item Is the IEP decidable for $\FOTOE$, that is, $\FOT$ with one equivalence relation? Note that the decidability and $\coNExpTime$-completeness proofs for $\FOTOE$ are significantly less involved than those for $\FOTTE$~\cite{Pratt23book}. What happens if we extend $\FOT$ with a single transitive (rather than equivalence) relation?

\item Is the $\FOTTE$-IEP still undecidable if one drops equality from $\FOT$? Note that, for $\FOT$ without equivalence relations and without equality, the IEP is decidable and the known complexity bounds are the same as for the case with equality~\cite{DBLP:conf/kr/KuruczWZ23}.

\item Is the IEP decidable for $\mathcal{ALCQIO}$? What about $\CT$, that is, $\FOT$ with counting?    
\end{itemize}



\providecommand{\noopsort}[1]{}

\end{document}